\documentclass[11pt]{article}
\usepackage[letterpaper,margin=1in]{geometry}

\usepackage{amssymb,stackrel,mathtools,complexity,url}

\usepackage{tikz}
\newcounter{braid}
\newcounter{strands}
\pgfkeyssetvalue{/tikz/braid height}{0.5cm}
\pgfkeyssetvalue{/tikz/braid width}{0.5cm}
\pgfkeyssetvalue{/tikz/braid start}{(0,0)}
\pgfkeyssetvalue{/tikz/braid colour}{black}
\pgfkeys{/tikz/strands/.code={\setcounter{strands}{#1}}}

\makeatletter
\def\cross{%
  \@ifnextchar^{\message{Got sup}\cross@sup}{\cross@sub}}

\def\cross@sup^#1_#2{\render@cross{#2}{#1}}

\def\cross@sub_#1{\@ifnextchar^{\cross@@sub{#1}}{\render@cross{#1}{1}}}

\def\cross@@sub#1^#2{\render@cross{#1}{#2}}

\def\render@cross#1#2{
  \def\strand{#1}
  \def\crossing{#2}
  \pgfmathsetmacro{\cross@y}{-\value{braid}*\braid@h}
  \pgfmathtruncatemacro{\nextstrand}{#1+1}
  \foreach \thread in {1,...,\value{strands}}
  {
    \pgfmathsetmacro{\strand@x}{\thread * \braid@w}
    \ifnum\thread=\strand
    \pgfmathsetmacro{\over@x}{\strand * \braid@w + .5*(1 - \crossing) * \braid@w}
    \pgfmathsetmacro{\under@x}{\strand * \braid@w + .5*(1 + \crossing) * \braid@w}
    \draw[braid] \pgfkeysvalueof{/tikz/braid start} +(\under@x pt,\cross@y pt) to[out=-90,in=90] +(\over@x pt,\cross@y pt -\braid@h);
    \draw[braid] \pgfkeysvalueof{/tikz/braid start} +(\over@x pt,\cross@y pt) to[out=-90,in=90] +(\under@x pt,\cross@y pt -\braid@h);
    \else
    \ifnum\thread=\nextstrand
    \else
     \draw[braid] \pgfkeysvalueof{/tikz/braid start} ++(\strand@x pt,\cross@y pt) -- ++(0,-\braid@h);
    \fi
   \fi
  }
  \stepcounter{braid}
}

\tikzset{braid/.style={double=\pgfkeysvalueof{/tikz/braid colour},double distance=1pt,line width=2pt,white}}

\newcommand{\braid}[2][]{%
  \begingroup
  \pgfkeys{/tikz/strands=2}
  \tikzset{#1}
  \pgfkeysgetvalue{/tikz/braid width}{\braid@w}
  \pgfkeysgetvalue{/tikz/braid height}{\braid@h}
  \setcounter{braid}{0}
  \let\sigma=\cross
  #2
  \endgroup
}
\makeatother

\newcommand{\Zed}{\mathbb{Z}}

\usepackage{amsmath,latexsym,amsmath,amsfonts,amsthm}
\usepackage{epsfig, epstopdf}

\newtheorem{definition}{Definition}
\newtheorem{lemma}{Lemma}
\newtheorem{corollary}{Corollary}
\newtheorem{theorem}{Theorem}

\def\vec#1{\mbox{\boldmath${#1}$}}

\def\Rom#1{\uppercase\expandafter{\romannumeral#1}}

\newcommand{\inv}[1]{\overline{#1}}
\newcommand{\ew}{\varepsilon}
\newcommand{\atilde}{\raise.17ex\hbox{$\scriptstyle\mathtt{\sim}$}}
\newcommand{\fgr}{\textup{FG}}

\newcommand{\sd}{\stackbin[]{*}{\Rightarrow}}
\newcommand{\rr}{\mathfrak{r}}

\newcommand{\rlem}[1]{Lemma~\ref{#1}}
\newcommand{\rcor}[1]{Corollary~\ref{#1}}

\newcommand{\req}[1]{Equation~(\ref{#1})}

\usetikzlibrary{spy}

\usetikzlibrary{intersections,automata}
\usetikzlibrary{arrows,calc}
\tikzset{
>=stealth',
help lines/.style={dashed, thick},
axis/.style={<->},
important line/.style={thick},
connection/.style={thick, dotted},
}

\usetikzlibrary{spy,shapes.geometric}

\tikzset{elliptic state/.style={draw,ellipse}}

\title{Composition problems for braids: \\Membership, Identity and Freeness}

\author{%
Sang-Ki Ko  \and Igor Potapov
  }

\begin{document}

\maketitle

\begin{abstract}
In this paper we investigate the decidability and complexity
of problems related to braid composition.
While all known problems for a class of braids with three strands, $B_3$,
have polynomial time solutions we prove that
a very natural question for braid composition, the membership problem,
is $\NP$-complete for braids with only three strands.
The membership problem is decidable in $\NP$ for $B_3$, but it becomes harder
for a class of braids with more strands. In particular
we show that fundamental problems about braid compositions
are undecidable for braids with at least five strands, but decidability of these problems for $B_4$ remains open.
Finally we show that the freeness problem for semigroups of braids from $B_3$ is also decidable in $\NP$.

The paper introduces a few challenging algorithmic problems about topological braids 
opening new connections between braid groups, combinatorics on words, complexity theory and 
provides solutions for some of these problems by application of several 
techniques from automata theory, matrix semigroups and algorithms.
\end{abstract}

\section{Introduction}
In this paper we investigate the decidability and complexity for
a number of problems related to braid composition. 
Braids are classical topological objects that attracted a lot of attention
due to their connections to topological knots and links 
as well as their applications to polymer chemistry, molecular biology,
cryptography, quantum computations and robotics \cite{DehornoyDRW08,EpsteinPCHLT92,LomonacoK11}.

The discovery of various cryptosystems based on the braid group 
inspired a new line of research about the complexity analysis
of decision problems for braids, including the word problem, 
the generalized word problem, root extraction problem, 
the conjugacy problem and the conjugacy search problem~\cite{Garber07,Garber10,Mahlburg04,MyasnikovSU05,MyasnikovSU06}.
For many problems the polynomial time solutions were found, but 
it was surprisingly shown by M. S. Paterson and A. A. Razborov in 1991 
that another closely related problem, the {\em non-minimal braid problem}, to be $\NP$-complete \cite{PatersonR91}.\\

{\bf Non-minimal Braid Problem:}
Given a word $\omega$ in the generators  $\sigma_1, \ldots , \sigma_{n-1}$  and their inverses, determine whether there is a 
shorter word $\omega'$   in the same  generators which represents the same element of
the $n$-strand braid group $B_n$?\sloppy\\
%
%
%

The main result of this paper is to show another $\NP$-hard problem
for braids in $B_3$, i.e. with only three strands.  
The problem can be naturally formulated in terms of composition (or concatenation)  of braids which is one of the fundamental operations for the braid group. 

Given two geometric braids, we can compose them, i.e. put one after the other making the endpoints of the first one coincide with the starting points of the second one. There is a neutral element for the composition: it is the trivial braid, also called identity braid, i.e. the class of the geometric braid where all the strings are straight. 
Two geometric braids are isotopic if there is a continuous deformation of the ambient space that
deforms one into the other, by a deformation that keeps every point in the two bordering planes fixed.
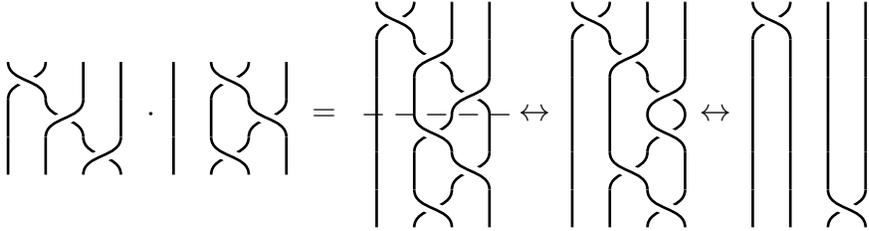
\begin{center}
\begin{tikzpicture}
\braid[strands=4,braid start={(-1.2,-0.8)}]
{\sigma_1 \sigma_2^{-1} \sigma_3^{-1}}
\node[font=\large] at (1.2,-1.5) {\(\cdot \)};
\braid[strands=4,braid start={(1,-0.8)}]
{\sigma_2 \sigma_3 \sigma_2}
\node[font=\large] at (3.5,-1.5) {\(= \)};
\braid[strands=4,braid start={(3.7,0)}]
{\sigma_1 \sigma_2^{-1} \sigma_3^{-1} \sigma_2 \sigma_3 \sigma_2 }
\node[font=\large] at (5,-1.5) {\( - - - - - \)};
\node[font=\large] at (6.3,-1.5) {\(\leftrightarrow \)};
\braid[strands=4,braid start={(6.3,0)}]
{\sigma_1 \sigma_2^{-1} \sigma_3^{-1} \sigma_3 \sigma_2 \sigma_3 }
\node[font=\large] at (8.7,-1.5) {\(\leftrightarrow \)};
\braid[strands=4,braid start={(8.7,0)}]
{\sigma_1  \sigma_5  \sigma_5 \sigma_5 \sigma_5 \sigma_3 }
\end{tikzpicture}
\end{center}
In this paper we study several computational problems related to composition of braids:
Given a set of braids with $n$ strands  $B=\beta_1, \dots ,\beta_k \in B_n$. Let us denote a semigroup of braids, generated by $B$ and
the operation of composition, by  $\langle B \rangle$.
\begin{itemize}
\item {\bf Membership Problem.}  Check whether exist a composition of braids from a set $B$  that is isotopic to a given braid $\beta$, i.e. is  $\beta$ in $\langle B \rangle$ ?
\item {\bf Identity Problem.} Check whether exist a composition of braids from a set $B$  that is isotopic to a trivial braid.
\item {\bf Group Problem.} Check whether for any braid $\beta \in {B}$  
 we can construct the inverse of $\beta$ by composition of braids from  $B$, i.e. is a semigroup $\langle B \rangle$ a group?
\item {\bf Freeness Problem.} 
Check whether any two different concatenations of braids from $B$ are not isotopic, i.e. is a semigroup of braids $\langle B \rangle$  free?
\end{itemize}
\begin{center}
  \begin{tabular}{ c | c | c  | c }
    \hline
     & $B_3$ & $B_4$ & $B_5$ \\ \hline
    Membership Problem & $\NP$-complete & ?  & Undecidable \\ 
    Group/Identity Problem & $\NP$ & ?  & Undecidable \\
    Freeness Problem& $\NP$ & ? & Undecidable \\
    \hline
  \end{tabular}
\end{center}

In contrast to many polynomial time problems we show
that the membership problem for $B_3$ is $\NP$-hard\footnote{Note that proposed $\NP$-hardness construction is not directly applicable for the identity problem in $B_3$.} by using a combination of new and
existing encoding techniques from automata theory, group theory, matrix semigroups \cite{BellHP12,BellP12} and algebraic properties of braids \cite{DehornoyDRW08}. 
Then we prove that the membership problem for $B_3$ is decidable in 
$\NP$, which is the first non-trivial case where composition is associative, but it is non-commutative. The main idea of the $\NP$ algorithm is to 
reduce the membership problem for $B_3$ into the emptiness problem 
for context-free valence grammars, which is already known to be an 
$\NP$-complete problem. Note that this improves the first 
decidability result shown in~\cite{Potapov13}.
The membership problem for braids in $B_3$ has a very close 
connection with other non-trivial computational problems in matrix semigroups.
For instance, the braid group~$B_3$ has a close relationship with the modular group~${\rm PSL}(2,\Zed)$ 
since the braid group $B_3$ is the universal central extension of ${\rm PSL}(2,\Zed)$. Recently, the problem of deciding whether a finitely generated 
matrix semigroup in ${\rm PSL}(2,\Zed)$ contains the identity matrix is proven 
to be $\NP$-complete~\cite{BellHP17}.
Note that the proposed $\NP$ algorithm for the membership problem
for $B_3$ was inspired by the work of several authors on
the membership problem for $2\times 2$ matrix semigroups 
\cite{BellHP12,BellP12,ChoffrutK10,GurevichS07}.
We also show that fundamental problems about the braid compositions
such as the identity and freeness problems
are undecidable for braids with at least five strands, but decidability of these problems for $B_4$ remains open.
It is worth mentioning that there is no embedding from a 
set of pairs of words into $B_4$~\cite{Akimenkov91}. 
Hence, these problems might be decidable for $B_4$ since our undecidability proofs for $B_5$ essentially rely on the embedding 
from a set of pairs of words into $B_5$.

Recently, there have been several papers on games on 
braids~\cite{BovykinC06,CarlucciDW11,HalavaHNP17} where one 
player called the {\em attacker} tries to reach the trivial braid and 
the other player called the {\em defender} tries to keep the attacker 
from reaching the trivial braid based on the composition of 
braids from a finite set. Halava et el.~\cite{HalavaHNP17} proved 
that it is undecidable to check for the existence of a winning 
strategy in $B_3$ from a given non-trivial braid and in $B_5$ 
from the trivial braid.

\section{Preliminaries}\label{notSec}
\subsection{Words and Automata} 
Given an alphabet $\Gamma = \{1, 2, \ldots, m\}$, a word $w$ is an element $w \in \Gamma^*$. We denote the concatenation of two words $u$ and $v$ by either $u\cdot v$ or $uv$ if there is no confusion.
For a letter $a \in \Gamma$, we denote by
$\inv{a}$ or $a^{-1}$ the inverse letter of $a$, such that $a\inv{a} = \varepsilon$ where $\varepsilon$ is the empty word. We also denote 
$\inv{\Gamma} = \Gamma^{-1} = \{\inv{1}, \inv{2}, \ldots, \inv{m}\}$ and for a word $w =w_1w_2\cdots w_n$, we denote 
$\inv{w} = w^{-1} = w_n^{-1} \cdots w_2^{-1}\, w_1^{-1}$. 


The free group over a generating set $H$ is denoted by $\fgr(H)$,
i.e., the free group over two elements $a$ and $b$ is denoted as
$\fgr(\{a,b\})$. 
For example, the elements of $\fgr(\{a,b\})$ are all the words
over the alphabet $\{a,b,a^{-1},b^{-1}\}$ that are reduced, i.e.,
that contain no subword of the form $x \cdot x^{-1}$ or $x^{-1}
\cdot x$ (for $x \in \{a,b\}$). Note that $x \cdot x^{-1} = x^{-1} \cdot x = \varepsilon$.

Let $\Sigma = \Gamma \cup \inv{\Gamma}$. Using the notation of \cite{AngPRS09}, we 
shall also introduce a reduction mapping which removes factors of the form $a\inv{a}$ for $a \in \Sigma$. 
To that end, we define the relation
$\vdash \,\subseteq \Sigma^* \times \Sigma^*$ such that for all $w, w' \in \Sigma^*$, $w \vdash w'$ if and only if there exists 
$u,v \in \Sigma^*$ and $a \in \Sigma$ where $w= ua\inv{a}v$ and $w' = uv$. We may then define by $\vdash^*$ the reflexive and transitive 
closure of $\vdash$. 

\begin{lemma}[\cite{AngPRS09}]\label{reducedLem}
For each $w \in \Sigma^*$ there exists exactly one word $r(w) \in \Sigma^*$ such that $w \vdash^* r(w)$ does not contain any factor of the
form $a\inv{a}$, with $a \in \Sigma$.
\end{lemma}

The word $r(w)$ is called the reduced representation of word $w \in \Sigma^*$. As an example, we see that if $w = 132\inv{2}1\inv{1}\,\inv{3}\,\inv{1} \in \Sigma^*$, then
$r(w) = \varepsilon$.
 
Using standard notations, a deterministic finite automaton (DFA) is given by quintuple $(Q, \Sigma', \delta, q_0, F)$ where $Q$ is the set of 
states, $\Sigma'$ is the \emph{input alphabet}, $\delta: Q \times \Sigma' \to Q$ is the \emph{transition function}, $q_0 \in Q$ is the initial
state and $F \subseteq Q$ is the set of final states of the automaton. We may extend $\delta$ in the usual way to have domain 
$Q \times \Sigma'^*$. Given a deterministic finite automaton $A$, the language recognized by $A$ is denoted by $L(A) \subseteq \Sigma'^*$, i.e.
for all $w \in L(A)$, it holds that $\delta(q_0, w) \in F$.

\begin{lemma}\label{power_automata}
For any given $n\in \mathbb{Z}, n \geq 3$ there is a DFA $P_n$ over a group alphabet $\Sigma$, $|\Sigma|=2n$,
with $n+2$ states and $2n$ edges such that the only  word  $w \in L(P_n)$ and $r(w) = \varepsilon$,
has length $|w| = 2^n$.
\end{lemma}
\begin{proof}
We adapt the proof of a related result over \emph{deterministic finite automata} (DFA) recently shown in \cite{AngPRS09}.
Define alphabets $\Gamma = \{1,2,\ldots,n\}$, $\inv{\Gamma} = \{\inv{1},\inv{2},\ldots,\inv{n}\}$ and 
$\Sigma = \Gamma \cup \inv{\Gamma}$. It is shown in \cite{AngPRS09} that for any $n \geq 3$, there exists a DFA
$A_n$, with $n+1$ states over $\Sigma$, such that for any
word $w \in \Sigma^*$ where $w \in L(A_n)$ and $r(w) = \varepsilon$ then $|w| \geq 2^{n-1}$. Their proof is 
constructive and we shall now show an adaption of it.
Let $Q = \{q_0, \ldots, q_{n+2}\}$ and $q_0$ be the initial state and $\{q_{n+2}\}$ is the final state. 
We define the transition function $\delta: Q \times \Sigma^* \to Q$ of the DFA such that:

$$
\delta(q_a, c) = \left\{ \begin{array}{ll} 	q_1, & \textrm{if } c = 1 \textrm{ and } a=0; \\
							q_{a+1}, & \textrm{if } c = \inv{a} \textrm{ and } 1 \leq a \leq n; \\
							q_0, & \textrm{if  } c = a \textrm{ and } 2 \leq a \leq n-1, \\
                                                        q_{n+2}, & \textrm{if  c = n} \textrm{ and } a=n+1; 
				\end{array} \right.
$$
All other transitions are not defined. The structure of this DFA can be seen in Figure~\ref{fig_2}. 
\begin{figure}[htb]
\begin{center}
\usetikzlibrary{positioning,automata}
\begin{tikzpicture}[->,>=stealth',shorten >=1pt,auto,node distance=2cm,
                    semithick]
  \node[elliptic state,initial]   (q_0)                {$q_0$};
  \node[elliptic state]           (q_1) at (2.5,0) {$q_1$};
  \node[elliptic state] (q_2) at (5,-1.2) {$q_2$};
  \node[elliptic state] (q_3) at (5, -3) {$q_3$};
  \node[elliptic state] (q_4) at (0, -3) {$q_{n-1}$};  
  \node[elliptic state,draw=none] (q_5) at (2.5, -3) {$\cdots$};    
  \node[elliptic state] (q_6) at (-2.5, -3) {$q_{n}$};    
  \node[elliptic state] (q_7) at (-5, -3) {$q_{n+1}$};      
  \node[accepting,elliptic state] (q_8) at (-5, -1) {$q_{n+2}$};        
  \path[->] (q_0) edge                node [above] {$1$} (q_1)
            (q_1) edge    [bend left] node [below] {$\inv{1}$} (q_2)
            (q_2) edge                node [right] {$\inv{2}$} (q_3)
            (q_2) edge                node [below] {$2$} (q_0)            
            (q_3) edge                node [below] {$3$} (q_0)
            (q_3) edge                node [below] {$\inv{3}$} (q_5)            
            (q_5) edge                node [below] {$\inv{n-2}$} (q_4)                        
            (q_4) edge                node [right] {$n-1$} (q_0)
            (q_4) edge                node [below] {$\inv{n-1}$} (q_6)
            (q_6) edge                node [below] {$\inv{n}$} (q_7)
            (q_7) edge                node [right] {$n$} (q_8);
\end{tikzpicture}
\caption{A deterministic finite automaton where the length of minimum non empty word $w$ such that $r(w) = \varepsilon$ is $2^{n}$.}\label{fig_2}
\end{center}
\end{figure}
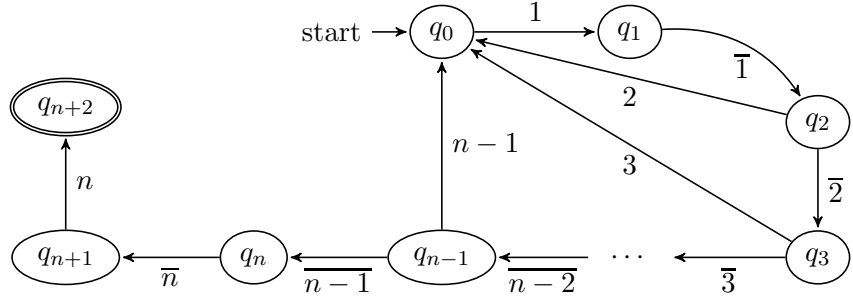
%
%
The only path leading to a state $q_n$, for any $n \geq 3$ 
with an empty reduced word  has length $2^n-2$. 
The path for reaching state $q_2$ with an empty reduced word has  length $2$
and there are no other paths leading to $q_2$ with an empty reduced word.
Let us assume that another path is leading to $q_2$ via a path where 
the larger index of a reachable state on this path is $j$. Then at least one symbol $j$
is not canceled in the reduced word leading to $q_2$. 
Consider a path from $q_i$ to $g_{i+1}$ which corresponds to reduced word $v$
then it should be of the form
$v=i \cdot u  \cdot \inv{i}$ where a word $u$ is an empty word
and it corresponds to a path from a state $q_0$ to $q_i$ otherwise
the reduced word of $v$ is not empty.

Let us assume that the path leading to a state $q_i$  
with an empty reduced word, i.e  $r(w)= \varepsilon$ has  length $2^i-2$.  Then
the path for reaching state $i+1$ with a reduced
word equal to the empty word  can be represented as a path
$w \cdot \inv{i} \cdot u {i}$ where $r(u) = \varepsilon$. 
Since $w$ is the only path to reach $q_i$ from $q_0$ then we have 
the required path has a form $w \cdot \inv{i} \cdot w {i}$
and its length is $(2^i-2) + 1+ (2^i-2)+1= 2^{i+1}-2$.
Finally we add two extra transitions to make the length of a 
path
to be $2^n$.
%
\end{proof}
\begin{lemma}\label{coding}
For any given $s\in \mathbb{Z}$ which has a binary representation of size m, i.e. $m=\lceil log_2(s) \rceil $,
there is a DFA $M_s$ over a group alphabet $\Sigma$, $|\Sigma|=O(m^2)$,
with  $O(m^2)$ states
such that the only  word $w \in L(M_s)$ and $r(w) = \varepsilon$, has a length $|w| = s$.
\end{lemma}
\begin{proof}
Let us represent $s$  as the following power series  
\begin{center}
${\alpha}_m 2 ^m + {\alpha}_{m-1} 2 ^{m-1} + \ldots + {\alpha}_{2} 2^{1} + {\alpha}_1 2 ^{0}$, 
where ${\alpha}_i \in \{0,1\}$.
\end{center}
For each non-zero ${\alpha}_i$ and $i \geq 3$ we will contract the automaton $P_i$ from Lemma~\ref{power_automata}
using unique non-intersecting alphabets for each automaton to avoid any possible cancellation of words between different parts of our final automaton.
Also for non-zero ${\alpha}_1$, ${\alpha}_2$ and ${\alpha}_3$ we define three different automata $P_1$, $P_2$, $P_3$ having 
a linear structure with one $\varepsilon$ transition, two consecutive $\varepsilon$ transitions and four 
consecutive $\varepsilon$ transitions, which will give us paths of length $2^0$, $2^1$ and $2^2$.

Then we will use a resulting set of automata $P_{i_1},  P_{i_2}, \ldots  P_{i_l}$ to build
a single automaton by merging the initial state of $P_{i_t}$ with the  final state of $P_{i_{t+1}}$ for all $t=1 \ldots l-1$
and defining the initial state of $P_{i_1}$ as the initial state of automaton $M_s$ and the final state of $P_{i_l}$
as the final state of $M_s$.
It is easy to see that following the Lemma~\ref{power_automata} each $P_{i_t}$ will reach its own final state having
an empty word iff the number of executed transition is $2^{i_t}$. So finally we build a DFA $M_s$ 
over a group alphabet, such that the only  word $w \in L(M_s)$ and $r(w) = \varepsilon$, has a length $|w| = s$.


The DFA $M_s$ over a group alphabet $\Sigma$,
will have  $|\Sigma|=O(m^2)$,  $O(m^2)$ states and $O(m^2)$ transitions,
 since  there are no more then $m$ parts $P_{i_1},  P_{i_2}, \ldots  P_{i_l}$ 
and each part  $P_{i_t}$  has only $i_t+2$ states.
Moreover the only  word $w \in L(M_s)$ and $r(w) = \varepsilon$, has a length $|w| = s$.
\end{proof}

\subsection{Context-Free Valence Grammar} 

A {\em (context-free) valence grammar} over $\mathbb{Z}^k$ is a context-free 
grammar in which every production has an associated value from $\mathbb{Z}^k$~\cite{FernauS02,Hoogeboom02}. A string in the language of the grammar can be derived in the usual way under the additional constraint that the sum of the associated values of the productions used in the derivation add up to $\vec{0} \in \mathbb{Z}^k$.

Formally, a valence grammar $G$ is specified as a quadruple~$(N,\Sigma,R,S),$ where $N$ is a set of {\em nonterminals}, $\Sigma$ is 
a set of {\em terminals}, $R \subseteq N \times (N \cup T)^* \times \mathbb{Z}^k$ is a set of {\em productions}, and $S \in N$ is the {\em axiom}. For an element~$(A, w, \vec{x}) \in R$, we write $A \xrightarrow{\vec{x}} w$, where $A \to w$ is the underlying production and $\vec{x} \in \mathbb{Z}^k$ is the associated value of the production.

Let $\alpha A \beta$ be a word over $N \cup \Sigma$, where $A \in N$ and $A \xrightarrow{\vec{x}} \gamma \in R$. 
Then, we say that A can be rewritten as $\gamma$ and 
the corresponding derivation step is denoted
 $(\alpha A \beta, \vec{r}) \Rightarrow (\alpha \gamma \beta, \vec{r} + \vec{x})$.
The reflexive, transitive closure of $\Rightarrow$ is denoted by $\sd$ and
 the (context-free) valence language generated
 by $G$ is  $L(G) = \{w \in \Sigma^* \mid(S,\vec{0}) \sd (w, \vec{0})\}$. 

\begin{lemma}\label{lem:valenceNP}
The emptiness problem for context-free valence grammars is $\NP$-complete.
\end{lemma}
\begin{proof}
It is known that the reachability problem in integer vector addition systems ($\mathbb{Z}$-VAS) is $\NP$-complete~\cite{ChistikovHH16}. Thus, $\NP$-hardness follows 
from the fact that a valence grammar~$G = (N,\Sigma, R, S)$ is a $\mathbb{Z}$-VAS if $R \subseteq N \times (N \cup \{ \varepsilon\}) \times \mathbb{Z}^k$ and $N = \{S\}$.

Moreover, the $\NP$ upper bound of the emptiness problem for 
context-free commutative grammars with integer counters ($\mathbb{Z}$-CFCGs)~\cite{ChistikovHH16} applies to the valence grammars since we can ignore the order of nonterminals and terminals when 
we consider the emptiness of the grammars. 
\end{proof}

\subsection{Braids} 
The braid groups can be defined in many ways including geometric, topological, algebraic and algebro-geometrical definitions \cite{Garside69}.
Here we provide algebraic definition of the braid group.
\begin{definition}\label{Braid_Artrin}
The $n$-strand braid group  $B_n$ is the group given by the presentation with $n-1$
generators $\sigma_1 , \ldots , \sigma_{n-1}$ and the following relations
$\sigma_i \sigma_j = \sigma_j \sigma_i$, for $|i-j| \geq 2$
and $\sigma_i \sigma_{i+1} \sigma_i = \sigma_{i+1} \sigma_i \sigma_{i+1}$
for $1 \leq i \leq n-2$. These relations are called Artrin's relation.
\end{definition}
\begin{definition}
Words in the alphabet $\{ {\sigma}$ , ${\sigma}^{-1} \}$ will be referred to as braid words
\footnote{ Whenever a crossing of strands $i$ and $i + 1$ is encountered, $\sigma_i$ or ${\sigma_i}^{-1}$  
is written down, depending on whether strand $i$ moves under or over strand $i + 1$.}.
\end{definition}

We say that a braid word $w$ is positive if no letter $\sigma_i^{-1}$ occurs in $w$. The positive braids form a semigroup denoted by $B_n^+$.
There is one very important positive braid known as the fundamental $n$-braid, $\Delta_n$.  
The fundamental braid of the group $B_n$ (also known as Garside element) can be written with $\frac{n(n-1)}{2}$ Artin generators as: 
$\Delta_n = (\sigma_{n-1} \sigma_{n-2} \ldots  \sigma_{1})(\sigma_{n-1} \sigma_{n-2} \ldots  \sigma_{2}) \ldots \sigma_{n-1}.$ 

Geometrically, the fundamental braid is obtained by lifting the bottom ends of the identity braid and flipping (right side over left) while keeping the ends of the strings in a line.  The inverse of the fundamental braid $\Delta_n$ is denoted by $\Delta_n^{-1}$.
\begin{center}
\begin{tikzpicture}
\node[font=\large] at (-1,-0.7) {\(\Delta\)};
\node[font=\large] at (-0.5,-0.7) {\(=\)};
\node[font=\large] at (0.0,-0.3) {\(\sigma_1\)};
\node[font=\large] at (0.0,-0.75) {\(\sigma_2\)};
\node[font=\large] at (0.0,-1.2) {\(\sigma_1\)};
\node[font=\large] at (3.8,-0.3) {\(\sigma_2\)};
\node[font=\large] at (3.8,-0.75) {\(\sigma_1\)};
\node[font=\large] at (3.8,-1.2) {\(\sigma_2\)};
\braid[strands=3,braid start={(-0.2,0)}]
{\sigma_1 \sigma_2 \sigma_1}
\node[font=\large] at (1.8,-0.7) {\(=\)};
\braid[strands=3,braid start={(1.8,0)}]
{\sigma_2 \sigma_1 \sigma_2}
\node[font=\large] at (4.8,-0.3) {\(\sigma_1\)};
\node[font=\large] at (4.8,-0.75) {\(\sigma_1^{-1}\)};
\node[font=\large] at (8.8,-0.3) {\(\sigma_2\)};
\node[font=\large] at (8.8,-0.75) {\(\sigma_2^{-1}\)};
\braid[strands=3,braid start={(4.8,0)}]
{\sigma_1  \sigma_1^{-1}}
\node[font=\large] at (6.8,-0.5) {\(=\)};
\braid[strands=3,braid start={(6.8,0)}]
{\sigma_2  \sigma_2^{-1}}
\node[font=\large] at (9.4,-0.5) {\(=\)};
\braid[strands=3,braid start={(9.3,0)}]
{\sigma_4  \sigma_4 }
\end{tikzpicture}
\end{center}

Let $B_3=\{\sigma_1,\sigma_2 | \sigma_1 \sigma_2 \sigma_1 = \sigma_2 \sigma_1 \sigma_2\}$  be the group with three braids.
Let $\Delta$ be the Garside element: $\Delta = \sigma_1 \sigma_2 \sigma_1$. Let $\tau: B_3 \to B_3$ be automorphism
defined by  $\sigma_1 \to \sigma_2$, $\sigma_2 \to \sigma_1$.  It is straightforward to check  that 
\begin{equation}\label{eq:tau}
\Delta \beta = \tau (\beta) \Delta,   \hspace{0.5cm} \Delta^{-1} \beta = \tau (\beta) \Delta^{-1},  \hspace{0.5cm} \beta \in B_3.
\end{equation}
\begin{lemma}[\cite{Orevkov08}]\label{Orevkov08}
Two positive words are equal in $B_3$ if and only if they can be obtained from each other by applying successively the relation 
$ \sigma_1 \sigma_2 \sigma_1 = \sigma_2 \sigma_1 \sigma_2\ $.
A positive word is left or right divisible by $\Delta$ if and only if it contains the subword 
$\sigma_1 \sigma_2 \sigma_1$ or  $\sigma_2 \sigma_1 \sigma_2\ $.
\end{lemma}
\begin{lemma}[Garside normal form~\cite{Garside69,DehornoyDRW08}]\label{lem:garside}
 - Every braid word $w \in  B_n$ can be written uniquely as
$\Delta^k \beta$, where $k$ is an integer and $\beta$ is a positive braid of which $\Delta$ is not a left divisor.
\end{lemma}
%
%

Two braids are isotopic if their braid words can be translated one into each other via the relations from the Definition~\ref{Braid_Artrin}
plus the relations $\sigma_i \sigma_i^{-1}  = \sigma_i^{-1} \sigma_i =1$, where $1$ is the identity (trivial braid).
%

Let us define a set of natural problems for semigroups and groups in the context of braid composition.
Given a finite set of braids $B$, a multiplicative semigroup $\langle B \rangle$ is a set of braids that can be generated by any finite composition of braids from $B$.

The {\em membership problem} asks,
given a braid  $\beta \in B_n$ and a finite set of braids $B \subseteq B_n$,
whether there exists a composition  $Y_{1}Y_{2}\cdots Y_{r}$, with each $Y_{i} \in B$ such that 
$Y_{1}Y_{2}\cdots Y_{r} = \beta$. In other words, is $\beta \in \langle B \rangle$?
In the membership problem, when braid $\beta$ is the trivial braid, we call this problem the {\em identity problem.} 
The identity problem for semigroups is a well-known challenging problem 
which is also computationally equivalent to another fundamental problem in
group theory called the {\em group problem}. The problem is, given a finitely generated semigroup $S$, to decide whether a subset of the generator of S generates a non-trivial group~\cite{ChoffrutK10}.
Finally, the {\em freeness problem} is to decide whether the given semigroup of braids $\langle B \rangle$ 
is free. 

%
%


\section{Membership Problem in the Braid Group~$B_3$ is $\NP$-complete }

In this section, we show that the membership problem for braids in the 
braid group~$B_3$ is $\NP$-complete. We also prove that the freeness problem 
in $B_3$ can be decided in $\NP$ based on the $\NP$ algorithm for the membership 
problem.

\subsection{$\NP$-hardness of the Membership Problem in $B_3$}

First we show that the membership problem is $\NP$-hard for braids in $B_3$.
Our reduction will use the the the subset sum problem which is a famous 
$\NP$-complete problem. In the subset sum problem, we are given a positive integer $x$ and a 
finite set of positive integer values $S = \{s_1, s_2, \ldots, s_k\}$ and asked whether 
there exists a nonempty subset of $S$ which sums to $x$.

We will require the following encoding between words over an arbitrary group alphabet and a binary group alphabet, which is well known from the literature.

\begin{lemma}\label{groupEnc} Let $\Sigma' = \{z_1, z_2, \ldots, z_l\}$ be a group alphabet and $\Sigma_2=\{c,d,\inv{c},\inv{d}\}$ be a binary 
group alphabet. Define the mapping $\alpha:\Sigma' \to \Sigma_2^*$ by:
\begin{center}
$\alpha(z_i) = c^i d\inv{c}^{i}, \alpha(\inv{z_i}) = c^i\inv{d}\inv{c}^{i}, $
\end{center}
where $1 \leq i \leq l$. Then $\alpha$ is a monomorphism~\footnote{A monomorphism is an injective homomorphism.} (see \cite{BirgetM08} for more details). Note that $\alpha$ can be extended to
domain $\Sigma'^*$ in the usual way.
\end{lemma}

\begin{lemma}[\cite{BezverkhnijD99}]\label{BraidsEnc}
Let $\Sigma_2=\{c,d,\inv{c},\inv{d}\}$ be a binary group alphabet and define $f: \Sigma_2^* \to {B}_3$ by:
$
f(c) = {{\sigma}_1}^4,
f(\inv{c}) = {{\sigma}_1}^{-4},
f(d) = {{\sigma}_2}^4,
f(\inv{d}) = {{\sigma}_2}^{-4}.
$
Then mapping $f$ is a monomorphism.
\end{lemma}

The above two morphisms give a way to map words from an arbitrary sized alphabet into the set 
braid words in $B_3$.
We will later require the following corollary concerning mappings $f$ and $\alpha$ to allow us to argue
about the size of braid words constructed by $f \circ \alpha$.

\begin{corollary}\label{matrixComp}
Let $\alpha$ and $f$ be mappings as defined in Lemma~\ref{groupEnc} and Lemma~\ref{BraidsEnc}, then:
$$f(\alpha(z_j)) = f(c^j d\inv{c}^{j}) = {{\sigma}_1}^{4j} {{\sigma}_2}^4 {{\sigma}_1}^{-4j}
$$
and the length of a braid word from $B_3$ corresponding to the symbol $z_j \in  \Sigma'$  is $8j+4$.
\end{corollary}
%
%
%
%
%
Now we prove that the membership problem for braid semigroups in $B_3$ 
is $\NP$-hard.

\begin{lemma}\label{lem:nphard}
The membership problem is $\NP$-hard for braids from $B_3$ 
\end{lemma}
\begin{proof}
We shall use an encoding of the {\em subset sum problem} (SSP) into a set of braids from $B_3$. Define an alphabet
$\Sigma = \Sigma' \cup \{ \Delta , \inv{\Delta} \}, \Sigma'= \{1, 2, \ldots, k+2, \inv{1}, \inv{2}, \ldots, \inv{k+2} \}$
that will be extended during the construction. 

We now define a set of words $W$ which will encode the SSP instance. Note that the length of words in the following set is not bounded by a polynomial of the size of the SSP instance, however this is only a transit step
and will not cause a problem in the final encoding. In particular the unary representation of a number $s$ by 
a word ${\Delta}^{2s}$ will be substituted by a set of words of a polynomial size of $i,j$ and $s$
that will generate a unique word  $i \cdot {\Delta}^{2s} \cdot j$.
$$W = 
\begin{array}{ll}
\{1 \cdot {\Delta}^{2s_1} \cdot \inv{2}, & 1 \cdot \ew \cdot \inv{2}, \\
2 \cdot {\Delta}^{2s_2} \cdot \inv{3}, & 2 \cdot \ew \cdot \inv{3}, \\
\vdots & \vdots \\
k \cdot {\Delta}^{2s_k} \cdot \inv{(k+1)}, \,\,\, & \,\,\, k \cdot \ew \cdot\inv{(k+1)}, \\
(k+1) \cdot \inv{{\Delta}}^{2x} \cdot \inv{(k+2)} \}  \subseteq \Sigma^* & \\
%
\end{array}
$$
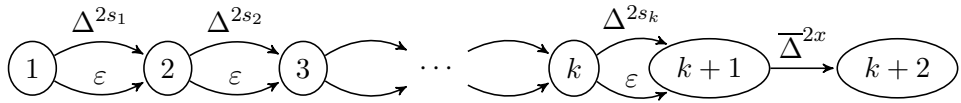
\begin{figure}[htb]
\begin{center}
\begin{tikzpicture}[->,>=stealth',shorten >=1pt,auto,node distance=1.8cm,
                    semithick]

  \node[elliptic state] (A)                  {1};
  \node[elliptic state]         (B) [right of=A] {2};
  \node[elliptic state]         (C) [right of=B] {3};
  \node[elliptic state,draw=none]         (C2) [right of=C] {$\cdots$};  
  \node[elliptic state]         (D) [right of=C2] {$k$};
    \node[elliptic state]         (E) [right of=D] {$k+1$};
        \node[elliptic state]         (F) at (11.5,0) {$k+2$};
  \path (A) edge     [bend left]         node {$\Delta^{2s_1}$} (B)
   edge     [bend right]         node {$\varepsilon$} (B)
(B) edge     [bend left]         node {$\Delta^{2s_2}$} (C)
   edge     [bend right]         node {$\varepsilon$} (C)
(C) edge     [bend left]         node {} (C2)
   edge     [bend right]         node {} (C2)   
(C2) edge     [bend left]         node {} (D)
   edge     [bend right]         node {} (D)      
(D) edge     [bend left]         node [above] {$\Delta^{2s_{k}}$}  (E)
   edge     [bend right]         node [above] {$\varepsilon$} (E)
(E) edge              node {$\inv{\Delta}^{2x}$} (F);
\end{tikzpicture}
\caption{The initial structure of a product which forms the identity on labels.}\label{fig_1}
\end{center}
\end{figure}
Figure~\ref{fig_1} shows the way in which the words of $W$ can be combined to give the identity for the reduced word on labels
in the graph structure. 
The above assumption will mean that
we start from node~$1$ of the graph and choose either $a^{s_1}$ or $\varepsilon$ to move to node~$2$. This corresponds to $w_1$ being equal
to either $1 \cdot {{\Delta}}^{2s_1} \cdot \inv{2}$ or $1 \cdot \ew \cdot \inv{2}$. We follow such nondeterministic choices from node~$1$ until we reach a node ${s_{k+2}}$.  At this point, if we chose $s_{i_1}, s_{i_2}, \ldots, s_{i_l}$, such that they sum to $x$, then the reduced representation of $w$
will equal $1 \cdot \inv{k+2}$.
If there does not exist a solution to the subset sum problem, then it will not be possible to reach the empty word 
concatenating the labels on a graph structure so it would be possible to get a word $1 \cdot \inv{k+2}$, since it will be only 
$1 \cdot  w' \cdot \inv{k+2}$, where $w' \neq \varepsilon$.  

Using the encoding idea from Lemma~\ref{coding} we replace 
each transition from state $i$ to state $j$ labelled with $ {\Delta}^{2s} j$ by the automaton $M_{2s}$
and then will encode each transition form  $M_{2s}$ from a state $x$ to state $y$ with the label $z \in \Sigma$
by the braid word $f(\alpha(x)) \cdot  (\sigma_1 \sigma_2 \sigma_3 )^2 \cdot  f(\alpha(z)) \cdot f(\alpha(\inv{y}))$
following Corollary~\ref{matrixComp}. We use $\Delta^2 =  (\sigma_1 \sigma_2 \sigma_3 )^2$  
rather then $\Delta$ to have unchanged structure of words
since $\Delta^2$ is commutative with any word in $B_3$.
Also each word of the following type 
$i \cdot \ew \cdot \inv{j}$, where $i,j \in \Sigma'$ can be directly encoded by a braid 
$f(\alpha(i)) \cdot f(\alpha(\inv{i}))$ .


The number of states, the alphabet size and the number of edges for each $M_{2s_i}$ automaton
are of the order $O(m^2)$, where $m$ is $log_2 s_i$. Thus we have that the whole automaton 
after replacing all $\Delta^{2 s_i}$ transitions by $M_{2s_i}$
will be encoded with the finite number of words of the order $O({k \cdot log^2 s})$, 
where $s$ is the maximal element of $\{ s_1, s_2, \ldots , s_k\}$ and 
the length of each braid word is of the order $O({k \cdot log^2 s})$.
In addition to that we add $k$ words representing $\varepsilon$ transitions.

Using Lemma~\ref{groupEnc}, we  encode the set of words $W$ into a set of braid words 
over the alphabet  $\{\sigma_1, \sigma_1^{-1}, \sigma_2, \sigma_2^{-1} \}$, where the total number
of letters will be only polynomially increased. So finally the SSP has a solution if and only if
the braid $f(\alpha(1)) \cdot f(\alpha(\inv{k+2}))$ belongs to the defined semigroup of braid words.
\end{proof}


\subsection{$\NP$ algorithm for the Membership Problem in $B_3$}

In this section, we show that the membership problem in $B_3$ is 
decidable in $\NP$. Note that the decidability of the membership 
problem has already been solved in \cite{Potapov13} but the 
time complexity of the proposed algorithm is exponential. 
The main idea of the algorithm proposed in~\cite{Potapov13} is as 
follows. Let us suppose that we are given a set~$B = \{\beta_1, \beta_2, \ldots, \beta_n\}$ of 
braid words and a braid word $\beta$ for which we need to decide whether $\beta$ 
can be generated by the set $B$. 
We first convert the given braids~$\beta_i$ for $1 \le i \le n$ into the unique Garside normal 
form $\Delta^{k_i} \beta_i^+ $ where $k_i \in \mathbb{N}$ is an integer 
and $\beta_i^+$ is a positive braid word by \rlem{lem:garside}.

First, we construct an automaton which accepts a regular language~$L_B = \{ \Delta^{k_i} \beta_i^+ \mid 1 \le i \le n \}^+$ 
over the alphabet~$\{ \sigma_1, \sigma_2, \Delta, \Delta^{-1}\}$ 
which is the set of non-empty products of braid words from $B$ in 
the Garside normal form. Then, we iteratively insert transitions 
labelled by any power of the fundamental braid~$\Delta$ whenever 
we find a sequence of transitions from the automaton 
corresponding to $\Delta^x$ for any $x \in \mathbb{N}$. As we may 
have cycles in the process of inserting transitions, some transitions 
are labelled by $\Delta^{Expr(x_1, x_2, \ldots, x_m)}$, where 
$expr(x_1, x_2, \ldots, x_m)$ is a linear expression over the variables 
$x_1, x_2 ,\ldots, x_m$. Now we solve the membership problem by nondeterministically choosing a sequence of transitions and 
solving the system of linear Diophantine equations that consist of 
the equations labeling the chosen transitions from the automaton. 
Hence we show that the problem is decidable but in exponential time 
as the construction of the automaton takes exponential time.

Here we tackle the membership problem in $B_3$ in a slightly different way to obtain the $\NP$ upper bound. We first construct a 
(context-free) valence grammar generating every braid word corresponding 
to the input braid~$\beta$ and compute the intersection of the 
valence grammar and the regular language $L_B$ which is the set of 
non-empty products of braid words from $B$. Now we can see
that the membership problem in $B_3$ can be reduced to the emptiness problem for context-free valence grammars in polynomial time
and it follow from \rlem{lem:valenceNP} that the 
membership problem in $B_3$ is also in $\NP$.
In the following we first prove that there exists a context-free 
valence grammar that generates the set of all braid words 
equal to the given braid~$\beta$.

%
%

\begin{lemma}\label{lem:valence}
Given a braid word~$\beta \in B_3$, there exists a context-free valence
grammar~$G$ over the alphabet~$\Sigma = \{\sigma_1, \sigma_2, \Delta, \Delta^{-1}\}$ such that 
$L(G)$ is the set of all braid words over $\Sigma$ which are equal to 
$\beta$.
\end{lemma}

\begin{proof}
First, we convert the given braid~$\beta$ into the Garside normal form~$ \beta^+ \Delta^k$. Note that $k \in \mathbb{Z}$ is an integer and $\beta^+ = \sigma_{i_1}\sigma_{i_2}\sigma_{i_3} \cdots \sigma_{i_n} \in B_3^+$ is a 
positive braid of length~$n$ where $i_j = \{1, 2 \}$ for $1 \le j \le n$.

We define a valence grammar~$G = (N,\Sigma, R,S)$, where 
\begin{itemize}
\item $N = \{S\} \cup \{S_1, S_2, \ldots, S_{n-1} \} \cup 
\{\overline{S_1}, \overline{S_2}, \ldots, \overline{S_{n-1}} \} \cup \{ A_{\sf even}, A_{\sf odd},A_{\sf any}\}$ is a finite set of nonterminals, 
\item $\Sigma = \{\sigma_1, \sigma_2 , \Delta, \Delta^{-1}\}$ is a set of 
terminals, 
\item $R$ is a finite set of productions, and
\item $S$ is the axiom.
\end{itemize}
We define 
$R$ to contain the following production rules:
\begin{itemize}
\item $S \xrightarrow{-k} A_{\sf even} \sigma_{i_1} S_1 \mid A_{\sf odd} \rr(\sigma_{i_1}) \overline{S_1}$,
\item $S_j \xrightarrow{0} A_{\sf even} \sigma_{i_{j+1}} S_{j+1} \mid A_{\sf odd} \rr(\sigma_{i_{j+1}})\overline{S_{j+1}}$ for $1 \le j \le n-2$,
\item $\overline{S_j} \xrightarrow{0} A_{\sf even} \rr(\sigma_{i_{j+1}}) \overline{S_j} \mid A_{\sf odd} \sigma_{i_{j+1}} S_{j+1}$ for $1 \le j \le n-2$,
\item $S_{n-1} \xrightarrow{0} A_{\sf even} \sigma_{i_n} A_{\sf any} \mid A_{\sf odd} \rr(\sigma_{i_n}) A_{\sf any} $,
\item $\overline{S_{n-1}} \xrightarrow{0} A_{\sf even} \rr(\sigma_{i_n}) A_{\sf any} \mid A_{\sf odd} \sigma_{i_n} A_{\sf any} $,
\item $A_{\sf even} \xrightarrow{0} \varepsilon \mid A_{\sf even}A_{\sf even} \mid A_{\sf odd} A_{\sf odd}$,
\item $A_{\sf even} \xrightarrow{1} \sigma_2 A_{\sf odd} \sigma_2 A_{\sf even} \sigma_1 \mid \sigma_1 A_{\sf even} \sigma_2 A_{\sf odd} \sigma_2 \mid \sigma_1 A_{\sf odd} \sigma_1 A_{\sf even} \sigma_2 \mid \sigma_2 A_{\sf even} \sigma_1 A_{\sf odd} \sigma_1,$\sloppy
\item $A_{\sf odd} \xrightarrow{1} \Delta,$
\item $A_{\sf odd} \xrightarrow{-1} \Delta^{-1},$
\item $A_{\sf odd} \xrightarrow{0} A_{\sf even} A_{\sf odd} \mid A_{\sf odd} A_{\sf even},$
\item $A_{\sf odd} \xrightarrow{1} \sigma_1 A_{\sf even} \sigma_2 A_{\sf even} \sigma_1 \mid \sigma_1 A_{\sf odd} \sigma_1 A_{\sf odd} \sigma_1 \mid \sigma_2 A_{\sf even} \sigma_1 A_{\sf even} \sigma_2 \mid \sigma_2 A_{\sf odd} \sigma_2 A_{\sf odd} \sigma_2,$ and\sloppy
\item $A_{\sf any} \xrightarrow{0} A_{\sf even} \mid A_{\sf odd}$.
\end{itemize}
Note that we can derive every braid word corresponding to 
the $\Delta^k$ where $k$ is an even (respectively, odd) integer from the nonterminal~$A_{\sf even}$ (respectively, $A_{\sf odd}$). In particular, 
the following derivation relation holds:
\[
(A_{\sf even}, 0) \sd (\omega, k),
\]
where $\omega$ is a braid word corresponding to $\Delta^k$ for an 
even integer~$k$. Similarly, $A_{\sf odd}$ can be replaced 
by every braid word corresponding to $\Delta^m$ where $m$ is 
an odd integer.

Now it remains to prove that the valence grammar~$G$ actually generates every braid word which is equal to 
the given braid~$\beta$ by the relations of the braid group~$B_3$.
First, we show that the every braid word generated by $G$ is equal to 
the given braid~$\beta$. Since the `$S$'-nonterminals should be substituted by 
regular type productions of $G$ (containing at most one `$S$'-nonterminal on the right-hand side) to derive words consisting of 
terminals, we see that the following derivation should 
be performed in any case:
\begin{equation}\label{eq:derivation}
(S, 0) \sd (A_0 \rr^{p(1)}(\sigma_{i_1}) A_1 \rr^{p(2)}(\sigma_{i_2}) A_2 \cdots A_{n-1}\rr^{p(n)} (\sigma_{i_n}) A_n, -k),
\end{equation}
where
\[
p(x)=
\begin{dcases}
0, & \mbox{if } |\{ k \mid 0 \le k < x,\;\; A_k = A_{\sf odd} \}| \equiv 0 \mod 2,\\
1, & \mbox{otherwise.}
\end{dcases}
\]
In other words, $p(x)$ has a value of 0 if the number of $A_{\sf odd}$ 
appearing in front of the $x$th terminal symbol of the 
positive braid~$\beta^+$ is even. Now we move the word generated by 
`$A$'-nonterminals to the right by~\req{eq:tau}. After moving every 
`$A$'-nonterminals to the right, we obtain the following braid word:
\[
\sigma_{i_1} \sigma_{i_2} \cdots \sigma_{i_n} A_0 A_1 \cdots A_n.
\]
It is easy to see that `$A$'-nonterminals can be substituted 
by some braid word which is equal to $\Delta^k$ as claimed above, 
and therefore, we prove that any braid word generated by $G$ is 
equal to the given braid~$\beta$.

Lastly, we show that the valence grammar~$G$ generates every braid 
word equal to the braid~$\beta$. First, we define the following 
two sets
\begin{itemize}
\item $C_{\sf even} = \{ \omega \mid \omega = \Delta^k, k \equiv 0 \mod 2\}$
and 
\item $C_{\sf odd} = \{ \omega \mid \omega = \Delta^k, k \equiv 1 \mod 2\}$
\end{itemize}
such that $C_{\sf even}$ (respectively, $C_{\sf odd}$) is the set of all braid words equal to 
the composition of an even (respectively, odd) number of the Garside element~$\Delta$.
Then, every braid word equal 
to $\beta$ is captured by the following set of words:
\[
C_0 \cdot \{\rr^{p'(1)}(\sigma_{i_1})\} \cdot C_1 \cdot \{\rr^{p'(2)}(\sigma_{i_2})\}\cdot C_2  \cdots C_{n-1} \cdot \{\rr^{p'(n)} (\sigma_{i_n}) \} \cdot C_n,
\]
where
\[
p'(x)=
\begin{dcases}
0, & \mbox{if } |\{ k \mid 0 \le k < x,\;\; C_k = C_{\sf odd} \}| \equiv 0 \mod 2,\\
1, & \mbox{otherwise.}
\end{dcases}
\]
Following the derivation described in (\ref{eq:derivation}), we can see that every braid word equal to $\beta$ can be derived by the valence grammar $G$.
\end{proof}

Now we are ready to present our $\NP$ algorithm for the membership problem in the 
braid group~$B_3$.

\begin{lemma}\label{lem:membershipNP}
The membership problem can be decided in $\NP$ for 
braids from $B_3$.
\end{lemma}

\begin{proof}
Let us suppose that we are given a set~$B = \{\beta_1, \beta_2, \ldots, \beta_n\}$ of 
braid words and a braid word $\beta$ for which we need to decide whether $\beta$ 
can be generated by the set $B$. 
We first convert the given braids~$\beta_i$ for $1 \le i \le n$ into the unique Garside normal 
form $\Delta^{k_i} \beta_i^+ $ where $k_i \in \mathbb{N}$ is an integer 
and $\beta_i^+$ is a positive braid word by \rlem{lem:garside}.

Let us define the
regular language $L_B = \{ \Delta^{k_i} \beta_i^+ \mid 1 \le i \le n \}^+$ 
over the alphabet~$\{ \sigma_1, \sigma_2, \Delta, \Delta^{-1}\}$.
Clearly, $L_B$ should contain a braid word~$\beta'$
which is equal to $\beta$ by the relations of the braid group~$B_3$ if and 
only if the given set~$B$ of braid words generates the target braid~$\beta$.

By \rlem{lem:valence}, there exists a valence grammar $G$ 
generating every braid word 
in $B_3$ equal to the target braid~$\beta$ and definable over 
the alphabet~$\{ \sigma_1, \sigma_2, \Delta, \Delta^{-1}\}$.
Therefore, the problem of checking whether $\beta$ can be generated 
by the set $B$ reduces to the problem of checking whether the intersection 
of $L_B$ and $L(G)$ is empty.

It is known that we can convert a given valence grammar over $\mathbb{Z}^k$ into a pushdown automaton (PDA) equipped with $k$ additional blind counters in polynomial time~\cite{FernauS02}. Since we are using an integer weight of dimension one, the valence grammar~$G$ can be converted into a PDA with a blind counter in 
polynomial time and construct a new PDA with a blind counter recognizing 
the intersection~$L_B \cap L(G)$ by constructing the Cartesian product 
of two automata. It is easy to see that the resulting automaton is still 
a PDA with a blind counter of size polynomial in the input size. By~\rlem{lem:valenceNP}, we conclude that the membership problem in the braid group~$B_3$ can be decided in $\NP$.
\end{proof}

Following \rlem{lem:nphard} and \rlem{lem:membershipNP}, we establish the following complexity result for the membership problem in the 
braid group~$B_3$.

\begin{theorem}\label{thm:npcomplete}
The membership problem for braids from the braid group~$B_3$ is $\NP$-complete.
\end{theorem}

\subsection{Freeness Problem in the Braid Group $B_3$}

In the proof of~\rlem{lem:membershipNP}, we construct a finite state automaton recognizing the 
regular language~$L_B$ with $n$ multi-states loops representing braid words in Garside normal form 
from the set $B$. Then, a path from the initial state to itself in this automaton represents a  
braid that can be constructed by a semigroup generator $ \{ \beta_1, \beta_2, \ldots ,\beta_n \} $.
Note that the rest of the proof is not based on the structure of this automata and
the same algorithm can be applied to check the membership for any other finite graph, where labels are braids from $B_3$. Hence, we can immediately establish the following result.

\begin{corollary}\label{cor:graph_membership}
Given a directed graph~$G$ with labels from the braid group~$B_3$, $u$ and $v$ are two nodes from $G$ and $\beta$ is a braid from $B_3$. 
Then, the problem of deciding whether exists a path $P$ from $u$ and $v$ such that a direct sum of braids on labels along a path $P$ is isotopic to a braid $\beta$ can be decided in $\NP$.
\end{corollary}

Note that the $\NP$ algorithm for the membership problem can exploited for decidability of the
freeness problem in braid semigroups with generators from $B_3$.

\begin{theorem}
The freeness problem for braids from the braid group~$B_3$ can be decided in $\NP$.
\end{theorem}

\begin{proof}
Let us consider a set $B=\{ \beta_1, \beta_2, \ldots \beta_n \}$ of braids from $B_3$ and a braid 
semigroup~$\langle B \rangle$ which is finitely generated by the set~$B$.
If the semigroup $\langle B \rangle$ is not free, then there are two products of the form  $A_1 \cdot X \cdot  A_2$
and $C_1 \cdot Y \cdot C_2$ such that 
\begin{equation}\label{eq:non-free}
A_1 \cdot X \cdot  A_2=C_1 \cdot Y \cdot C_2,
\end{equation}
where $A_1 \ne C_1$, $A_2 \ne C_2$,  $A_1,A_2,C_1,C_2 \in B$, and $X,Y \in  \langle B \rangle$.

Now it is not difficult to see that we can check whether the semigroup $\langle B \rangle$ is free 
if we can decide whether there exist two products as in~\req{eq:non-free} since we can iteratively 
run the same procedure for each pairs of braids from the set $B$. Indeed, we can decide 
whether there exist two products as in~\req{eq:non-free} for the chosen braids 
$A_1,A_2,C_1,C_2$ from the set $B$ by checking whether the following equation can 
be satisfied for some $X, Y \in \langle B \rangle$:
\[A_1 X A_2  {C_2}^{-1} Y^{-1} {C_1}^{-1} = I.\]
Then, we can construct a finite-state automaton recognizing all 
the sequences of braids of the form on the left-hand side of the equation and further construct 
an automaton that recognizes the following regular language over braids from $B_3$:
\[
\begin{split}
L_B = \{ A_1 w_1 A_2 C_2^{-1} w_2^{-1} C_1^{-1} \mid  A_1 \ne C_1, \;\; A_2 \ne C_2, \;\; A_1,A_2,C_1,C_2 \in B,\;\;\\ w_1, w_2 \in B^* \}.
\end{split}
\]
It should be noted that the construction of the automaton recognizing $L_B$ takes 
polynomial time. We can see that the braid semigroup~$\langle B \rangle$ is not free if and 
only if the regular language~$L_B$ contains any braid word corresponding to the trivial 
braid, which can be checked in $\NP$ by~\rcor{cor:graph_membership}. 
Hence, we conclude that the freeness problem for braid semigroups in $B_3$ can be decided in $\NP$.
\end{proof}
\section{Undecidability of Decision Problems in the Braid Group~$B_5$} 
The composition problems become harder with a larger number of strands. 
Since the braids group~$B_5$ contain the direct product 
of two free groups, it is possible to show that most of the composition problems are undecidable in $B_5$.
We first provide the following property of $B_5$ which will be used later in the undecidability results in 
$B_5$.

\begin{lemma}[\cite{BezverkhnijD99}]\label{product}
Subgroups $\langle {\sigma_1}^4 ,  {\sigma_2}^4  \rangle$, $\langle {\sigma_4}^2 ,  d  \rangle$ of the group $B_5$ are free
and $B_5$  contains the direct product  
$\langle {\sigma_1}^4 ,  {\sigma_2}^4  \rangle \times \langle {\sigma_4}^2 ,  d  \rangle$ of two free groups of rang 2
 as a subgroup, where $d = \sigma_4 \sigma_3 \sigma_2 \sigma_1^2 \sigma_2 \sigma_3 \sigma_4$.
\end{lemma}

We can prove the undecidability of the identity problem and the group problem by relying on 
the embedding from $B_5$ into the direct product of two free groups.

\begin{theorem}
The identity problem and the group problem are undecidable for braids in $B_5$. 
\end{theorem}
\begin{proof}

Bell and Potapov~\cite{BellP10} has proven the undecidability of the {\em identity correspondence problem} (ICP)
which asks whether a finite set of pairs of words (over a group alphabet) 
can generate an identity pair by a sequence of concatenations. 
Let $\Sigma = \{a,b\}$ be a binary alphabet and 
$\Pi = \{(s_1, t_1), (s_2, t_2), \ldots, (s_m, t_m)\} \subseteq \fgr(\Sigma) \times \fgr(\Sigma).$
Formally speaking, the ICP is to determine if there exists a nonempty finite sequence of
indices $l_1, l_2, \ldots, l_k$ where $1 \leq l_i \leq m$ such that 
$
s_{l_1}s_{l_2}\cdots s_{l_k} = t_{l_1}t_{l_2} \cdots t_{l_k} = \varepsilon,
$
where $\varepsilon$ is the empty word (identity).

We can directly use the \rlem{product} to encode the ICP in terms of braid words.
We shall use a straightforward encoding to embed an instance of the
ICP into a set of braids.
Let $\Pi \subseteq \Sigma^* \times \Sigma^*$ be an instance of the ICP
where $\Sigma = \{a,b,a^{-1},b^{-1}\}$ generates a free group. 
Define two morphisms $\phi$ and $\psi$ that map $\Sigma$ into $B_5$ as follows:
\[
\begin{aligned}
\phi(a) &= {\sigma_1}^4, & \phi(a^{-1}) &= {\sigma_1}^{-4} , \\
\phi(b) &= {\sigma_2}^4, &\phi(b^{-1}) &= {\sigma_2}^{-4}. \\
\psi(a) &= {\sigma_4}^2, & \psi(a^{-1}) &= {\sigma_4}^{-2}, \\
\psi(b) &=  \sigma_4 \sigma_3 \sigma_2 \sigma_1^2 \sigma_2 \sigma_3 \sigma_4, &
\psi(b^{-1}) &= \sigma_4^{-1} \sigma_3^{-1} \sigma_2^{-1} \sigma_1^{-2} \sigma_2^{-1} \sigma_3^{-1} \sigma_4^{-1}.
\end{aligned}
\]
The domain of $\phi$ and $\psi$ can be naturally extended to words as follows:
$$\phi (w_1  \ldots w_i)= \phi (w_1)  \cdot \ldots \cdot \phi (w_i); \hspace{0.3cm}
\psi (v_1  \ldots v_j)= \psi (v_1) \cdot  \ldots \cdot \psi (v_j),$$
where $w_1 \cdots w_i, v_1 \cdots v_j \in \Sigma^*$.
For each pair of words $(s,t) \in W$, define the
braid word $\phi(s) \cdot \psi(t)$. 
Let $S$ be a braid semigroup
generated by these braid words. In other words, $S$ is finitely generated 
by the set~$\{ \phi(s) \cdot \psi(t) \mid (s,t) \in \Pi\}.$
If there exists a solution to the ICP, then we
see that $\phi(\varepsilon) \cdot \psi(\varepsilon) = 1 \in S,$
where $1$ is the trivial braid. Otherwise, the trivial braid does not exist in the 
braid semigroup~$S$ since
$\psi$ and $\phi$ are injective homomorphisms.
Therefore, we have that the problem whether a trivial braid can be expressed by any finite length composition of braids from $B_5$ is undecidable. 

The identity problem is also computationally equivalent to the following problem which is called 
the group problem. Given a semigroup generated by a finite set of pairs of words (over a group alphabet), 
can we decide whether the semigroup is a group?
Using the same morphisms $\phi$ and $\psi$, we can encode the group problem for words by braids, having 
that the group problem for braids in $B_5$ is also undecidable.
\end{proof}

Similarly, we also prove that the freeness problem is undecidable in $B_5$.

\begin{theorem}
The freeness problem for braids from the braid group~$B_5$ is undecidable.
\end{theorem}
\begin{proof}
We first introduce the {\em mixed modification PCP} (MMPCP)~\cite{CassaigneKH96} 
which is already proven to be undecidable 
and prove the undecidability of the freeness problem in $B_5$
by encoding an instance of the MMPCP.

Given a finite alphabet $\Sigma$, a binary alphabet $\Delta$, and a pair of homomorphisms 
$h,g : \Sigma^* \to \Delta^*$, the MMPCP asks to decide whether or not there exists 
a word~$w = a_1\ldots a_k \in \Sigma^+, a_i \in \Sigma$ such that 
\[
h_1(a_1) h_2(a_2) \ldots h_k(a_k) = g_1(a_1) g_2 (a_2) \ldots g_k(a_k),
\]
where $h_i, g_i \in \{ h,g\}$ and for some $j \in [1,k]$ such that $h_j \ne g_j$.

Let $\Sigma = \{a_1, a_2, \ldots, a_{n-2}\}$ and $\Delta = \{ a_{n-1}, a_n\}$ be disjoint 
alphabets and $h,g : \Sigma^* \to \Delta^*$ be an instance of the MMPCP.
Now define a morphism~$\gamma : (\Sigma \cup \Delta)^* \times (\Sigma \cup \Delta)^*  
\to B_5$ by
\[
\gamma(u,v) = \phi(u) \cdot \psi(v).
\]
It is easy to see that $\gamma$ is a homomorphism since $\gamma(u_1, v_1) \gamma(u_2,v_2) = 
\gamma(u_1u_2, v_1v_2)$. Now let $S$ be a braid semigroup which is finitely generated by 
the set~$ \{ \gamma(a_i, h(a_i)), \gamma(a_i, g(a_i)) \mid a_i \in \Sigma, 1 \le i \le n-2\}$. 
The braid semigroup~$S$ is not free if and only if the MMPCP instance has a solution. 
Since the MMPCP is undecidable, we conclude that the freeness problem in the braid group~$B_5$ 
is also undecidable.
\end{proof}

\section{Conclusion}
The paper introduces a few challenging algorithmic problems about topological braids 
opening new connections between braid groups, combinatorics on words, complexity theory and 
provides solutions for some of these problems by application of several 
techniques from automata theory, matrix semigroups and algorithms.

We have shown that the membership problem for $B_3$ is decidable and actually $\NP$-complete. 
The $\NP$-hardness result is in line with the best current knowledge about similar problem in 
the special linear group~${\rm SL}(2,\Zed)$. W
Finally in this paper we have proven that fundamental problems about the braid compositions
are undecidable for braids with at least $5$ strands, but decidability of these problems for $B_4$ remains open.
\section*{Acknowledgements}
The second author is grateful for many fruitful discussions with Sergei Chmutov and Victor Goryunov on the computational problems in topology.

\bibliographystyle{abbrv} 
\bibliography{braid}

\begin{thebibliography}{10}

\bibitem{Akimenkov91}
A.~M. Akimenkov.
\newblock Subgroups of the braid group ${B_4}$.
\newblock {\em Mathematical notes of the Academy of Sciences of the USSR},
  50(6):1211--1218, 1991.

\bibitem{AngPRS09}
T.~Ang, G.~Pighizzini, N.~Rampersad, and J.~Shallit.
\newblock Automata and reduced words in the free group.
\newblock {\em CoRR}, abs/0910.4555, 2009.

\bibitem{BellHP12}
P.~C. Bell, M.~Hirvensalo, and I.~Potapov.
\newblock Mortality for 2 $\times$ 2 matrices is np-hard.
\newblock In {\em Proceedings of the 37th International Symposium on
  Mathematical Foundations of Computer Science}, MFCS 2012, pages 148--159,
  2012.

\bibitem{BellHP17}
P.~C. Bell, M.~Hirvensalo, and I.~Potapov.
\newblock The identity problem for matrix semigroups in {SL}$(2,\mathbb{Z})$ is
  {NP}-complete.
\newblock In {\em Proceedings of ACM-SIAM Symposium on Discrete Algorithms
  2017}, SODA 2017, pages 187--206, 2017.

\bibitem{BellP10}
P.~C. Bell and I.~Potapov.
\newblock On the undecidability of the identity correspondence problem and its
  applications for word and matrix semigroups.
\newblock {\em International Journal of Foundations of Computer Science},
  21(6):963--978, 2010.

\bibitem{BellP12}
P.~C. Bell and I.~Potapov.
\newblock On the computational complexity of matrix semigroup problems.
\newblock {\em Fundamenta Infomaticae}, 116(1-4):1--13, 2012.

\bibitem{BezverkhnijD99}
V.~{Bezverkhnij} and I.~{Dobrynina}.
\newblock {On the unsolvability of the conjugacy problem for subgroups of the
  group ${R_5}$ of pure braids.}
\newblock {\em {Mathematical Notes}}, 65(1):13--19, 1999.

\bibitem{BirgetM08}
J.-C. Birget and S.~W. Margolis.
\newblock Two-letter group codes that preserve aperiodicity of inverse finite
  automata.
\newblock {\em Semigroup Forum}, 76(1):159--168, 2008.

\bibitem{BovykinC06}
A.~Bovykin and L.~Carlucci.
\newblock Long games on braids, 2006.
\newblock Available online at
  \url{http://logic.pdmi.ras.ru/~andrey/braids_final3.pdf}.

\bibitem{CarlucciDW11}
L.~Carlucci, P.~Dehornoy, and A.~Weiermann.
\newblock Unprovability results involving braids.
\newblock {\em Proceedings of the London Mathematical Society},
  102(1):159--192, 2011.

\bibitem{CassaigneKH96}
J.~Cassaigne, J.~Karhum\"{a}ki, and T.~Harju.
\newblock On the decidability of the freeness of matrix semigroups.
\newblock Technical report, Turku Center for Computer Science, 1996.

\bibitem{ChistikovHH16}
D.~Chistikov, C.~Haase, and S.~Halfon.
\newblock Context-free commutative grammars with integer counters and resets.
\newblock {\em Theoretical Computer Science}, 2016.
\newblock In press.

\bibitem{ChoffrutK10}
C.~Choffrut and J.~Karhum\"{a}ki.
\newblock Some decision problems on integer matrices.
\newblock {\em RAIRO - Theoretical Informatics and Applications},
  39(1):125--131, 3 2010.

\bibitem{DehornoyDRW08}
P.~Dehornoy, I.~Dynnikov, D.~Rolfsen, and B.~Wiest.
\newblock {\em Ordering Braids}.
\newblock Mathematical surveys and monographs. American Mathematical Society,
  2008.

\bibitem{EpsteinPCHLT92}
D.~B.~A. Epstein, M.~S. Paterson, J.~W. Cannon, D.~F. Holt, S.~V. Levy, and
  W.~P. Thurston.
\newblock {\em Word Processing in Groups}.
\newblock A. K. Peters, Ltd., 1992.

\bibitem{FernauS02}
H.~Fernau and R.~Stiebe.
\newblock Sequential grammars and automata with valances.
\newblock {\em Theoretical Computer Science}, 276(1-2):377--405, 2002.

\bibitem{Garber07}
D.~Garber.
\newblock Braid group cryptography.
\newblock {\em CoRR}, abs/0711.3941, 2007.

\bibitem{Garber10}
D.~Garber.
\newblock Braid group cryptography.
\newblock In {\em Braids: Introductory Lectures on Braids, Configurations and
  Their Applications}, pages 329--403. World Scientific Publishing Company,
  2010.

\bibitem{Garside69}
F.~A. Garside.
\newblock The braid group and other groups.
\newblock {\em The Quarterly Journal of Mathematics}, 20(1):235--254, 1969.

\bibitem{GurevichS07}
Y.~Gurevich and P.~Schupp.
\newblock Membership problem for the modular group.
\newblock {\em SIAM Journal on Computing}, 37(2):425--459, 2007.

\bibitem{HalavaHNP17}
V.~Halava, T.~Harju, R.~Niskanen, and I.~Potapov.
\newblock Weighted automata on infinite words in the context of
  attacker-defender games.
\newblock {\em Information and Computation}, 2017.
\newblock Submitted.

\bibitem{Hoogeboom02}
H.~J. Hoogeboom.
\newblock Context-free valence grammars - revisited.
\newblock In {\em Proceedings of the 5th International Conference on
  Developments in Language Theory}, pages 293--303, 2002.

\bibitem{LomonacoK11}
S.~J. Lomonaco and L.~H. Kauffman.
\newblock Quantizing braids and other mathematical structures: the general
  quantization procedure.
\newblock In {\em Proceedings of SPIE - The International Society for Optical
  Engineering}, volume 8057, pages 805702--805702--14, 2011.

\bibitem{Mahlburg04}
K.~Mahlburg.
\newblock As overview of braid group cryptography, 2004.
\newblock Available online at
  \url{http://www.math.wisc.edu/~boston/mahlburg.pdf}.

\bibitem{MyasnikovSU05}
A.~Myasnikov, V.~Shpilrain, and A.~Ushakov.
\newblock A practical attack on a braid group based cryptographic protocol.
\newblock In {\em Proceedings of the 25th Annual International Cryptology
  Conference Advances in Cryptology}, CRYPTO 2005, pages 86--96, 2005.

\bibitem{MyasnikovSU06}
A.~Myasnikov, V.~Shpilrain, and A.~Ushakov.
\newblock Random subgroups of braid groups: An approach to cryptanalysis of a
  braid group based cryptographic protocol.
\newblock In {\em Proceedings of the 9th International Conference on Theory and
  Practice in Public-Key Cryptography}, PKC 2006, pages 302--314, 2006.

\bibitem{Orevkov08}
S.~Orevkov.
\newblock Quasipositivity problem for 3-braids.
\newblock {\em Turkish Journal of Mathematics}, 28(1):89--94, 2004.

\bibitem{PatersonR91}
M.~Paterson and A.~Razborov.
\newblock The set of minimal braids is co-np-complete.
\newblock {\em Journal of Algorithms}, 12(3):393--408, 1991.

\bibitem{Potapov13}
I.~Potapov.
\newblock {Composition Problems for Braids}.
\newblock In {\em Proceedings of the 33rd IARCS Annual Conference on
  Foundations of Software Technology and Theoretical Computer Science}, FSTTCS
  2013, pages 175--187, 2013.

\end{thebibliography}

\end{document}